\documentclass{article}
\usepackage{cmap}
\usepackage{mathtext}
\usepackage[T2A]{fontenc}
\usepackage[cp1251]{inputenc}

\usepackage{indentfirst}
\usepackage{misccorr}
\usepackage[font=small, labelsep=period]{caption}
\usepackage{cite}
\usepackage{array}
\usepackage{amsmath}
\usepackage{amssymb}
\usepackage{graphicx}
\usepackage{nicefrac}
\usepackage{wrapfig}
\usepackage{fancybox}
\usepackage{multirow}

\usepackage[pdftex, left=1in, right=1in, top=1in, bottom=2cm]{geometry}

\usepackage{tikz}
\usepackage{amsthm}
\newtheorem{theorem}{Theorem}

\newtheorem{definition}{Definition}
\newtheorem{condition}{Condition}

\newcommand{\kw}[1]{\medskip\par\textbf{#1:}~}
\newcommand{\kwe}{\kw{Key words}}

\newcommand{\up}[1][i]{\ensuremath^{(#1)}}
\newcommand{\F}{{\cal F}}
\newcommand{\Fq}{\F_q}
\newcommand{\cc}{ {\ensuremath{ {\cal C}} } }
\newcommand{\Ri}{R_I}              % rate of identification
\newcommand{\nw}{n_w}            % length of word w 
\newcommand{\rand}{\ensuremath{\mathrm{rand}}} 
\newcommand{\Smes}{{\cal U}}  % set of messages (identifiers)
\title{Identification based on random coding\thanks{This is the English translation of the paper \cite{SD22} published  in ``Proceedings of Moscow Institute for Physics and Technology'' in Russian.}}
\author{V.\,R. Sidorenko$^{1,2}$, C. Deppe$^{1}$}
\date{$^1$Technical University of Munich\\
 $^2$Institute for Information Transmission Problems (Kharkevich Institute) RAS}
\begin{document}
\maketitle
\begin{center}
 \emph{Dedicated to the memory of Ernst M. Gabidulin (1937 --2021)}   
\end{center}

\medskip
\begin{abstract}
	Ahlswede and Dueck showed possibility to identify with high probability one out of $M$ messages by transmitting $1/C\log\log M$ bits only, where $C$ is the channel capacity. It is known that this  identification can be based on error-correcting codes. We propose an identification procedure based on random codes that  achieves  channel capacity. Then we show that this procedure can be simplified using pseudo-random generators.  
 
%	Ahlswede and Dueck show possibility to identify with high probability one out of $M$ messages by transmitting $1/C\log\log M$ bits only, where $Ñ$ is the channel capacity. An optimal identification, achieving  channel capacity, can be based on error-correcting codes. We propose an optimal identification, based on random codes, and it's simplification using pseudo-random generators.  	
\kwe identification, random codes, pseudo-random generators.

\end{abstract}
\bigskip

 \section{Introduction}
 The standard model of telecommunication is as follows. There is a set  $\Smes = \{u\}$ that consists of  $M$ messages $u$. Every message has length $k=\log_2 M$ bit and can be selected with equal probability for transmission via a channel. To transmit a message  $u$ we need to transmit a packet of $k$ bits using the channel  $$1/C\log_2 M$$ times, where the channel capacity is $C$ bits per channel use. The transmission of long enough messages can be organized in a way that they will be delivered correct almost for sure. 
 
 \medskip
 In  1989 Ahlswede and Dueck \cite{AD} proposed to consider another model called  \emph{message identification}. This model assumes that the receiver expects just one message $u$, which is interesting for him. The transmitter does not know which message is interesting for the receiver.  After receiving a packet, the receiver should decide whether the expected message $u$ was transmitted or not.  Decision errors of the first and second type are allowed.  Ahlswede and Dueck \cite{AD} showed, that for arbitrarily small probabilities of errors in decision making, there exists an identification procedure which for a sufficiently long message requires using the channel
  $$1/C\log_2 \log_2 M$$ times only. It is said that such a procedure reaches the identification capacity and is called optimal.  The gain is achieved due to the fact that for identification it is not necessary to transmit the entire message of length $\log_2 M$,  but it is enough to transmit only its short ``tag'' of length $\log_2 \log_2 M$. 
 
 In addition, \cite{AD} shows that, without loss of efficiency, the identification problem over a noisy channel can be divided into the identification problem (over a noiseless channel) and the problem of transmitting a message tag over a noisy channel. In this regard, we will deal with the problem of identification itself, i.e., we assume that the channel is noiseless with capacity $C = 1$. 
 
 \medskip
 \emph{Interest in identification} has been revived with the advent of the concept of the Internet of Things (IoT). Let's assume that each of the huge number $M$ of IoT objects has its own identifier $u\in \Smes$. The central node, in order to establish a connection with the object $u$, sends a tag of the identifier $u$ over the network. With the help of the identification procedure, the object $u$ accepts the call, the other objects reject the call with a high probability. In this example, authentication can significantly reduce the length of the transmitted request. In addition, the identification problem is closely related to the \cite{BB} authentication problem used in practice, as well as other \cite{MK} problems, which also arouses interest in identification.
 
\medskip
Although the possibilities of identification are shown theoretically in \cite{AD}, \emph{the problem of constructing practically realizable identification procedures for messages of finite length} with acceptable complexity remains open. In the works \cite{AD}, \cite{VW}, \cite{MK} it is proposed to build an identification procedure based on error-correcting codes. A constructive Verd\'u-Wei identification procedure based on cascading two Reed-Solomon codes was proposed by \cite{VW} and discussed in \cite{SDF}, \cite{GKSS}, \cite{FTBDLMV} and others. The  Verd\'u-Wei procedure is asymptotically optimal. However, \cite{SDF} shows that the complexity of the  Verd\'u-Wei procedure is very high for practical applications, i.e., for messages of finite length, since it requires working with finite fields of large sizes. The possibility of using Reed-Muller codes for identification was studied in \cite{SFD}.
 
 \medskip
In this paper, we propose an optimal identification procedure based on random linear codes. It is shown that this procedure can be simplified by using pseudo-random sequence generators.

\section{Code based identification}

We will consider a linear $(n,k,d)_q$ code $\cc$ over a finite field $\Fq$, where $q $ is a prime power, of length $n$, dimension $k$, having code distance $d$ in the Hamming metric. Denote the generator matrix of the code by
$$G=(g\up[1],g\up[2],\dots,g\up[n]),$$
where $g\up[i]$ denote the $i$th column of the matrix. Then the codewords are the vectors $c = (c_1,\dots,c_n) = uG$, where $u\in \Fq^k$ are all possible message vectors.

\subsection{Identification procedure}
\emph{Transmission.} In order to transmit the identifier (message) $u\in\{1,\dots,M\}$, the encoder computes the word $c=uG \in \cal C$, then chooses the random index $i\in[1,n]$ uniformly and transmits through the channel the \emph{identifier word}
\begin{equation}
w = w_i = (i,c_i),
\end{equation}
consisting of a randomly chosen index $i$ and the $i$th \emph{tag $c_i$ of the identifier $u$}. Length of the identification word $w$ in bits is
\begin{equation}
\nw = \log n + \log q.
\end{equation}
Here and below, $\log = \log_2$ means logarithm to base 2.

\medskip

\emph{Decoder} receives the word $w = (i,c_i)$ over the noiseless channel and must decide whether the expected message $u'$  was transmitted or not. To do this, the decoder calculates the word $c'=u'G \in \cal C$ and compares its component $c_i'$ with the received tag $c_i$. If $c_i' = c_i$, then the decoder decides that \emph{the expected message was received}, i.e. that $u_i' = u_i$. Otherwise, i.e., if $c_i' \ne c_i'$, $u_i' \ne u_i$ is considered and \emph{message is rejected}.

Probability of type I decoder error, i.e., the probability of rejecting (missing) the expected message, is denoted by $\lambda_1$,
and the probability of an error of the type II, i.e., the probability of accepting a false message is denoted by $\lambda_2$.

\subsection{Analysis of identification procedure}
From the described code-based identification procedure, it immediately follows that in the case of a noiseless channel, the probability of missing is $\lambda_1 = 0$. 

Note that neither the encoder nor the decoder needs \emph{to decode the $\cal C$ code in the usual sense}. It is not even necessary to completely encode the message $u$ into the codeword $c$, it is required to calculate only one symbol $c_i$, i.e., the scalar product $c_i =(u,g\up[i])$.

The probability $\lambda_2$ of false acceptance is estimated as follows. The words $c$ and $c'$ of the $(n,k,d)_q$-code  $\cal C$ differ at least in $d$ positions $i$, but coincide, i.e., $c_i' = c_i$ in $n-d$ positions at most. Since false acceptance occurs when $c_i' = c_i$ matches, we get that for any pair of messages $u'\ne u$
\begin{equation}\label{lambda2}
\lambda_2 \le \frac{n-d}{n} = 1-\frac{d}{n}.
\end{equation}  

\emph{Identification rate} is defined as 
\begin{equation}
\Ri = \frac{\log \log |\Smes|}{\nw} =
 \frac{\log k + \log \log q }{\log n + \log q}.
\end{equation}

The false acceptance probability  $\lambda_2$ can be reduced \cite{FTBDLMV} by sending several, say $\ell$, identification words $w_{i_1}, w_{i_2},\dots,w_{i_\ell}$. The decoder accepts a message if it is accepted for each of the $\ell$ words $w_{i_j}$, $j=1,\dots,\ell$. In this case, the new false acception probability is upper bounded by $$\lambda_2^\ell \le \left(1-\frac{d}{n}\right)^\ell.$$

The identification procedure using the code sequence $\cc$ achieves the identification capacity \cite{AD} and is called \emph{optimal}
if the following  conditions are satisfied with the corresponding increase in the code parameters \cite{VW}:

\begin{equation}\label{q/n}
\frac{\log q}{\log n} \longrightarrow 0,
\end{equation}

\begin{equation}\label{k/n}
\frac{\log k}{\log n} \longrightarrow 1,
\end{equation}
meaning that for a noiseless channel the identification rate tends to 1, $\Ri \longrightarrow 1$, and also
\begin{equation}\label{d/n}
\frac{d}{n} \longrightarrow 1,
\end{equation}
i.e., the probability of false acceptance of the message $\lambda_2 \longrightarrow 0$. 
\begin{definition}
	A sequence of $(n,k,d)_q$-codes $\cc$ satisfying the conditions \eqref{q/n} - \eqref{d/n} is called an optimal identification sequence of codes.   
\end{definition}

Recall that for an integer $q\ge 2$ and $x\in [0,1]$ \emph{the $q$-ary entropy function} $h_q(x)$ is defined as
\begin{equation}
h_q(x) = x\log_q(q-1) - x\log_q(x) - (1-x)\log_q(1-x).
\end{equation}

It was stated in \cite{VW}  that using a $q$-ary $(n, Rn, \delta n)_q$ code $\cc$ that asymptotically as $n\longrightarrow \infty$ reaches the Varshamov-Gilbert bound
\begin{equation}\label{VG}
R = k/n \le 1 - h_q(\delta)
\end{equation}
allows to get an optimal identification procedure. The next section suggests using a random code to build such a code $\cc$.

\section{Using a random code for identification}

Let us construct a random linear $(n,k,d)_q$ code $\cc$ over the field $\F_q$. To do this, we choose a random generator matrix of the code
\begin{equation}\label{G_rand}
G= \left(g_j\up[i]\right)
\end{equation}
choosing elements $g_j\up[i]$ of the matrix $G$ at random, independently and uniformly from the field $\F_q$. It is known that the random code $\cc$ asymptotically reaches the Varshamov-Gilbert bound with high probability.
\begin{theorem} 
	Let $\delta\in [0,1-1/q)$ and $0<\epsilon<1-h_q(\delta)$. Then
	with probability $P=1-q^{-\epsilon n +1}$ a random $(n,k,d)_q$ code $\cc$ has rate
	\begin{equation}\label{R_random}
	R= k/n \ge 1 - h_q(\delta) - \epsilon,
	\end{equation}
	and the relative code distance $d/n$ is not less than $\delta$.
\end{theorem}

The proof can be found, for example, in \cite{Gur}. 

\begin{theorem}\label{Rand_opt}
	The identification procedure based on the random code $\cc$ is asymptotically optimal with probability 1, i.e. almost for sure.
\end{theorem}
%\selectlanguage{english}
\begin{proof}
For $q\longrightarrow \infty$ the Varshamov-Gilbert bound \eqref{VG} approaches the Singleton bound $R=1-\delta$. It follows from \eqref{R_random}  that for any $\epsilon >0$  the rate $R$ of a random code of sufficiently large lengths $n$ with probability $P=1-q^{-\epsilon n +1}$ is at least $R\ge 1 - \delta -\epsilon - \gamma(q)$, where $\gamma(q) \longrightarrow 0$ as $q$ grows. Since the parameters of the random code are the same as the parameters of the code in \cite[Theorem~3]{VW}, our proof coincides with the proof in \cite[Theorem~3]{VW}, where the choice of growing parameters of the code sequence that ensures that conditions \eqref{q/n} - \eqref{d/n} are satisfied was given. The parameter $\epsilon$ is chosen so that $\epsilon \longrightarrow 0$, but $P=1-q^{-\epsilon n +1}\longrightarrow 1$, which completes the proof.
\end{proof}	
%\selectlanguage{russian}

Theorem~\ref{Rand_opt} shows the asymptotic optimality of an identification procedure based on random codes. However, for practical applications, the behavior of this procedure for finite parameter values is of interest. Let us show by an example that this procedure can also be of interest for finite values of the code parameters.

\medskip
{\bf Example 1}. Let 

$$q = 2^{10} = 1024,$$

$$n = q*2^8 = 262144,$$ 

$$\delta = 1 - 2^{-7},$$

$$\epsilon = 2^{-14}.$$

%Tîãäà ïî Òåîðåìå 1 ñ âåðîÿòíîñòüþ $$P = 1- 10^{-45}$$ ñëó÷àéíûé $q$-è÷íûé  $(n,k,d)$ êîä $\cc$ èìååò ðàçìåðíîñòü $k=340$ è ðàññòîÿíèå, êàê ìèíèìóì $\delta n$. Ýòî äàåò ïðîöåäóðó èäåíòèôèêàöèè  ñ $$|\Smes| = 2^{3400}$$ ñîîáùåíèÿìè, ñ âåðîÿòíîñòüþ ëîæíîãî ïðèåìà $$\lambda_2 = 1/128 = 0.0078$$ è ñî ñêîðîñòüþ èäåíòèôèêàöèè  $$\Ri = 0.419.$$ 
%
%Ñêîðîñòü èäåíòèôèêàöèè ñîñòàâëÿåò òîëüêî 42\% îò ïðîïóñêíîé ñïîñîáíîñòè $C=1$ èç-çà îòíîñèòåëüíî íåáîëüøèõ çíà÷åíèé ïàðàìåòðîâ. Çàòî, âû÷èñëåíèÿ â ïîëå $\F_{2^{10}}$ íå ïðåäñòàâëÿþò áîëüøîãî òðóäà. Íî äàæå ïðè òàêîé ñêîðîñòè $\Ri$,  äëÿ èäåíòèôèêàöèè ñîîáùåíèÿ ïî êàíàëó  ïåðåäàåòñÿ ñëîâî $w$ äëèíû ëèøü $\nw = 28$ áèò â òî âðåìÿ êàê äëÿ \emph{ïåðåäà÷è} ñîîáùåíèÿ ïîòðåáîâàëîñü áû ïåðåäàòü 3400 áèò. Ïëàòîé çà ýòîò âûèãðûø ÿâëÿåòñÿ íåíóëåâàÿ âåðîÿòíîñòü ëîæíîãî ïðèåìà $\lambda_2$. Êàê óêàçûâàëîñü, ýòà âåðîÿòíîñòü ìîæåò áûòü óìåíüøåíà  äî $\lambda_2^{\ell}$, åñëè ïðåäàâàòü $\nw \ell = 28\ell$ áèò \cite{FTBDLMV}. Òàê, ïðè $\ell =2$ ïðèäåòñÿ ïåðåäàòü 56 áèò, íî  âåðîÿòíîñòü ëîæíîãî ïðèåìà óìåíüøèòñÿ äî $0.00006$.

Then by Theorem 1 with probability $$P = 1- 10^{-45}$$ a random $q$-ary $(n,k,d)$ code $\cc$ has dimension $k=340$ and distance, at least $\delta n$. This gives the identification procedure with $$|\Smes| = 2^{3400}$$ messages, with false acceptance probability $$\lambda_2 = 1/128 = 0.0078$$ and identification rate $$\Ri = 0.419.$$

The identification rate is only 42\% of the capacity $C=1$ due to the relatively small values of the parameters. On the other hand, calculations in the  field $\F_{2^{10}}$ are not too complicated. But even at this rate $\Ri$, a word $w$ of length only $\nw = 28$ bits is transmitted over the channel to identify the message, while \emph{transmission} of the message would require to transmit  3400 bits. The payoff for this gain is the non-zero probability of false acceptance $\lambda_2$. As stated, this probability can be reduced to $\lambda_2^{\ell}$ by transmitting $\nw \ell = 28\ell$ bits \cite{FTBDLMV}. E.g., for $\ell =2$, 56 bits will have to be transmitted, but the probability of a false reception will decrease to $0.00006$.

\section{Using a pseudo-random number generator (PRNG)}

With the direct use of a random code, the disadvantage is the need to store the generator matrix by both the encoder and the decoder. This shortcoming can be avoided if each time before the transmission to calculate the column of the generator matrix using a pseudo-random number generator (PRNG).

A PRNG is a finite state machine having $\mu$ $q$-ary memory elements. Given the initial state (seed) $\sigma \in \Fq^\mu$, the generator generates a sequence $s = \rand_n(\sigma)$ of the required length $n$ of field elements. The period of the generated sequence cannot exceed $q^\mu$. We arrive at the following \emph{$(q,k,\ell,\mu)$-PRNG identification} procedure.

\subsection{$(q,k,\ell,\mu)$-PRNG-identification}
\medskip\emph{Transmission.}
 To send the message $u$, we choose a random initial state $\sigma$, generate $\ell$ random column vectors $g\up[1],g\up[2],\dots,g\up[\ell] \in \Fq^k$ of length $k$ one by one using the PRNG   and form the matrix
\begin{equation}\label{G}
G(\sigma)=(g\up[1],g\up[2],\dots,g\up[\ell]).
\end{equation}
The matrix $G(\sigma)$ can be considered as a submatrix consisting of $\ell$ columns of the generator matrix \eqref{G_rand} of the random code.

Compute the vector $\tau\in \Fq^\ell$ consisting of $\ell$ tags $\tau_l = u_l g_l$ of message $u$
\begin{equation}\label{compute_tau}
\tau = uG(\sigma).
\end{equation}
The following identification word is transmitted over the channel
\begin{equation}\label{w}
w = (\sigma, \tau).
\end{equation}
The word $w$ of length $\mu + \ell$ $q$-ary symbols is obtained by concatenating a random PRNG state $\sigma$ and a vector $\tau$ consisting of $\ell$ tags.

\medskip
\emph{Decoding}. The decoder receives the word $w$ transmitted over the noiseless channel and must decide whether the expected message $u'$   was transmitted. The decoder computes the matrix $G(\sigma)$ \eqref{G} using a PRNG with the initial state $\sigma$ received over the channel. Then the  vector of tags of the expected message is calculated
$$\tau' = u'G(\sigma).$$
If $\tau'= \tau$, the decoder accepts the message as the expected one. When $u' = u$, the decoder will definitely accept the message, hence $\lambda_1=0$. If $\tau'\ne \tau$, then the decoder rejects the message. A false acceptance occurs when $u' \ne u$ but $\tau'= \tau$. Denote by $P_2 (u,u')$ the probability of false acceptance for a pair of messages $u,u'$. We want the condition $P_2 (u,u') \le \lambda_2$ to be satisfied for all pairs $u\ne u'$. 

\medskip
The optimality of the method depends on the PRNG parameters. The main question is: what should be the memory $\mu$ to generate sequences of length $\ell k$ with properties close to random.

\medskip
\emph{Rate} of a $(q,k,\ell,\mu)$-PRNG-identification is
\begin{equation}
\Ri = \frac{\log \log |\Smes|}{\nw} =
\frac{\log k + \log \log q }{(\mu + \ell)\log q}.
\end{equation}

\medskip\medskip
\emph{Probability $\lambda_2$ of false message acceptance}.
We will assume that the following condition is satisfied.

\medskip
\begin{condition}  
	The PRNG generates a sequence $s = \rand_n(\sigma)$ of sufficient length $n$ of independent uniformly distributed random variables $s_i\in \Fq$.
\end{condition}

\begin{theorem}\label{T_lambda2}
	 Let the decoder expecting the message $u'$ receive the identification word $w =(\sigma, \tau)$ corresponding to the message $u$. Under Condition 1, for any $u\ne u'$, the probability of false acceptance $P_2 (u,u') = q^{-\ell}$,
	 i.e.,
	 \begin{equation}\label{lambda2_bound}
	 \lambda_2=q^{-\ell}.
	 \end{equation}
	\end{theorem}
%\selectlanguage{english}
\begin{proof}
	False acceptance of the message $u$ instead of the expected message $u'\ne u$ occurs if and only if the tags of these messages match, i.e. when equality holds
	$$u'G(\sigma) = uG(\sigma),$$
	which is equivalent to the equality
	\begin{equation}\label{vG}
	vG(\sigma)=0,
	\end{equation}
	where $v=u'-u \ne 0$, $v\in\Fq^k$. By Condition 1, the matrix $G(\sigma)$ consists of elements of the field $\Fq$ chosen independently and uniformly. The equality \eqref{vG} holds if and only if each column $g\up[i]$ of the matrix $G(\sigma)$ satisfies
	\begin{equation}\label{eq_vg}
	vg\up[i]=0 \quad i=1,\dots,\ell,
	\end{equation}
	what happens with probability $q^{-\ell}$.
\end{proof}
%\selectlanguage{russian}

%\medskip
%\emph{Optimality}. 
%%\begin{theorem}
%	The proposed  approach is asymptotically optimal under Assumption 1.
%%\end{theorem}
%%\begin{proof}
%	Let the field size $q$ is fixed while $k\rightarrow \infty$ and $\ell\rightarrow \infty$ such that 
%	\begin{equation}
%	\frac{\ell+\log \ell}{\log k} \rightarrow 0.
%	\end{equation}
%	Then from \eqref{R} add Lemma \ref{L1} it follows that the rate $R\rightarrow 1$ and the  probability of false acceptance $\lambda_2 \rightarrow 0.$ 
%%\end{proof}

\subsection{Requirements for the generator PRNG}
One of the standard tools for constructing pseudo-random sequences are linear feedback shift registers (LFSRs). LFSRs are easy to implement, fast to run, and easy to analyze. In this section, we explore the possibility of using LFSR as a PRNG for identification.

\medskip
\emph{Linear Feedback Shift Register (LFSR)} is shown in Fig. 1. The register consists of $\mu$ memory cells that store the elements of the  field $\Fq$, where $\mu$ is the length of the register. The initial state of the register is denoted by
\begin{equation}
\sigma = (s_1,s_2,\dots,s_\mu).
\end{equation}
On the $i$th cycle, $i=1,2,\dots$, a new element $s_{\mu+i}$
\begin{equation}\label{prng}
s_{\mu+i} = -\sum\limits_{j=1}^{\mu} a_j s_{i+\mu-j}
\end{equation}
 of the generated sequence is calculated using the feedback,
the content of the register is shifted to the left, generating an output element $s_i$, and a new element $s_{\mu+i}$ is written to the last cell of the register that has become free. The polynomial $a(x) = 1+a_1 x + a_2 x^2 +\dots + a_\mu x^\mu$, $a_0 = 1$, over the field $\Fq$ is called the feedback polynomial. An LFSR of length $\mu$ with feedback polynomial $a(x)$ is denoted by $\mu,a(x)$.

%A linear feedback shift register (LFSR) of \emph{length $\mu$} is shown in Fig. 1. It generates  sequences $s$ that satisfy 
%
%Where $\lambda(x) = 1+a_1 x + a_2 x^2 +\dots + a_\mu x^\mu$ is a connection polynomial over $\Fq$, with $a_0 = 1$. The state $\sigma$ of the LFSR is the contents of the register, i.e., $\sigma$ is a vector of length $\mu$. 

%\selectlanguage{english}
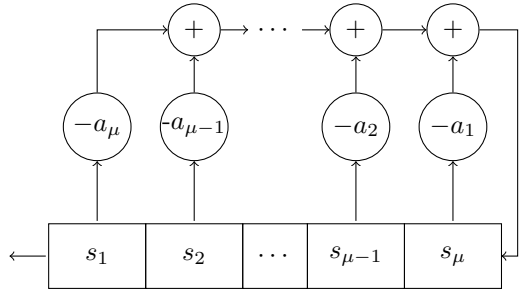
\begin{figure}[ht]
	\begin{center}
		\begin{tikzpicture} [scale=0.43,coeff/.style={circle,draw,minimum size=0.9cm, inner sep=0pt},coeff1/.style={circle,draw,minimum size=0.9cm, inner sep=0pt},coeff3/.style={circle,draw,minimum size=0.2cm}]
		
		% %draw rectangulars
		\foreach \x in {1,4,9,12}
		{
			\draw(\x,1) +(-1.5,-1) rectangle +(1.5,1);
			%\draw[->](\x,2)--(\x,3);
			
		}
		\draw (5.5,0) rectangle +(2,2);
		
		\draw (1,1) 	node 	(rect1)	{$ s_{1}$};
		\draw (4,1) 	node 	(rect2)	{$ s_{2}$};
		\draw (9,1) 	node 	(rect3)	{$ s_{\mu-1}$};
		\draw (12,1) 	node 	(rect4)	{$ s_{\mu}$};
		\draw (6.5,1) 	node	(rect5)		{$ \cdots$};	
		\draw (6.5,8) 	node	(rect6)		{$ \cdots$};
		
		%list point nodes
		\node (p1) at (1,1.8) {};
		\node (p2) at (4,1.8) {};
		\node (p3) at (9,1.8) {};
		\node (p4) at (12,1.8) {};
		
		\node(p5p) at (-0.3,1){};
		\node (p5) at (-2,1) {};
		\node (p6p) at (13.3,1){};
		\node (p6) at (15,1){};
		
		%draw the layer \theta
		%		\node[coeff] (coeff11) at (1,4) {};
		%		\node[coeff] (coeff12) at (4,4) {$\theta^{\mu-1}$};
		%		\node[coeff] (coeff13) at (9,4) {$\theta^2$};
		%		\node[coeff] (coeff14) at (12,4) {$\theta^1$};
		
		%draw the layer a	
		\node[coeff1] (coeff21) at (1,5) {$-a_\mu$};
		\node[coeff1] (coeff22) at (4,5) {$\textrm{-}a_{\mu-1}$};
		\node[coeff1] (coeff23) at (9,5) {$-a_2$};
		\node[coeff1] (coeff24) at (12,5) {$-a_1$};
		
		%draw the layer addition		
		\node[coeff3] (coeff32) at (4,8) {$+$};
		\node[coeff3] (coeff33) at (9,8) {$+$};
		\node[coeff3] (coeff34) at (12,8) {$+$};

		\draw [->]	(p1)--(coeff21); 	
		\draw [->]	(p2)--(coeff22); 
		\draw [->]	(p3)--(coeff23); 
		\draw [->]	(p4)--(coeff24); 
		
		% %draw arrow line	
		\draw [->](p5p)--(p5);		
		\draw [->] (coeff32.east)--(rect6.west);
		\draw [->] (rect6.east)--(coeff33.west);
		\draw [->](coeff34.east)--(14,8)--(14,1)--(13.5,1);
		
		%		\draw [->] (coeff11) --(coeff21);
		%		\draw [->] (coeff12) --(coeff22);
		%		\draw [->] (coeff13) --(coeff23);
		%		\draw [->] (coeff14) --(coeff24);
		
		\draw [->] (coeff21) --(1,8)--(coeff32);
		\draw [->] (coeff22) --(coeff32);
		\draw [->] (coeff23) --(coeff33);
		\draw [->] (coeff24) --(coeff34);
		\draw [->] (coeff33)--(coeff34);	
		\end{tikzpicture}
		\caption{Linear Feedback Shift Register  $\mu,a(x)$}
		\label{fig:shift_reg}
	\end{center}
\end{figure}
%\selectlanguage{russian}

\begin{theorem}\label{th_mu}
 To satisfy \eqref{lambda2_bound} for $(q,k,\ell,\mu)$-PRNG identification, the length $\mu$ of the LFSR linear feedback register must be at least $k$, i.e., $\mu\ge k$.
%ëèèíåéíàÿ ñëîæíîñòü ëþáîé áåñêîíå÷íîé ïîñëåäîâàòåëüíîñòè  $s$  íà âûõîäå ãåíåðàòîðà PRNG áûëà áîëüøå $k$, ò.å. $\varkappa(s) > k$.	
\end{theorem}
%\selectlanguage{english}
\begin{proof}
Assume the opposite, i.e., that $\mu < k$. For any sequence $s$ generated by the LFSR $\mu,a(x)$ with an arbitrary initial state $\sigma$, and for any $i\ge \mu+1$, holds the equality
\begin{equation}\label{recurs}
\sum\limits_{j=0}^{\mu} a_j s_{i-j}=0.
\end{equation}
As in the proof of Theorem~\ref{T_lambda2}, we denote the difference of two messages by the vector $v=u-u'\ne 0$ of length $k$. Since the messages can be arbitrary, the vector $v\ne 0$ can be chosen arbitrarily. We construct the vector $v$ as follows: take the vector
\begin{equation}
(a_\mu, a_{\mu-1}, \dots, a_1, 1)
\end{equation}
of length $\mu +1\le k$ and, if necessary, append it with zeros up to the length of $k$. Since the columns of the matrix $G$ in \eqref{G} are segments of the sequence $s$, using \eqref{recurs} we get that the conditions \eqref{eq_vg} and \eqref{vG} of false message acceptance are satisfied. Since for the chosen message pair $u,u'$ a false positive occurs for any choice of $\sigma$, the false positive probability $\lambda_2 =1$ and \eqref{lambda2_bound} does not hold.
\end{proof}
%\selectlanguage{russian}

Theorem~\ref{th_mu} shows that it is inexpedient to use LFSR as a PRNG in the considered scheme. Indeed, the length of the LFSR must be at least $k$, but then the identification word \eqref{w} would have to be transmitted with a length of more than $k$ field elements. This does not make sense, since the transmission of the entire message over a noiseless channel requires only $k$ field elements to be transmitted.

This situation arose due to the fact that both the pseudo-random generator \eqref{prng} and the procedure for calculating the tag \eqref{compute_tau} are linear. The   difficulty can be avoided by using a non-linear pseudo-random generator or a non-linear tag calculation function. Although an LFSR cannot be directly used as a PRNG, it is known that combining multiple LFSR generators allows the construction of non-linear PRNGs. The question of the choice of PRNG and the parameters of PRNG identification remains open for further research.

\section{Conclusion}
It is known in coding theory that linear random codes are ``good'' with high probability, that is, they reach the Varshamov-Gilbert bound. However, in practice, random codes are not used to transmit messages. The reason is that efficient decoding procedures are not known for these codes. The specifics of the task of identification based on codes is that it is not required to decode the codes!

This allowed us to propose an identification procedure based on random linear codes. The procedure is optimal, i.e. asymptotically with the growth of parameters allows to achieve the identification capacity of the channel. The given example shows that the proposed procedure with finite values of the code parameters can be interesting for practical applications.

In addition, it is shown that this procedure can be simplified using pseudo-random sequence generators, but it is not possible to directly use a linear feedback shift register (LFSR) for this. Either a non-linear generator is needed, possibly obtained by combining several LFSRs, or the scalar product used to calculate a tag must be replaced by a non-linear function.

The work leaves open a number of questions, the main of which is the choice of a nonlinear generator with small memory and high linear complexity. This may serve as a topic for further research.

\hfil \hbox to 0.3\textwidth{\hrulefill} \smallskip

\section*{Acknowledgement}
This work was supported by the German Federal Ministry of Education and Research (BMBF), Grant 16KIS1005.
Christian Deppe acknowledge the financial support by the Federal Ministry of Education and Research of Germany in the programme of ''Souver\"an. Digital. Vernetzt''. Joint project 6G-life, project identification number: 16KISK002.

%%References
%\selectlanguage{english}
%\begin{thebibliography}{99}
%\selectlanguage{russian}
%
%\bibitem{s1} \textit{Shevchenko D.\,V., Shevchenko V.\,P.} Selection and optimization of constructing autonomous seismic detection struction a landmark type // Proceedings of the VIII Russian scientific conference Modern security technology and means of complex security objectives.~--- 2010. --- P.~128--133. --- (in Russian).
%
%\bibitem{s2} \textit{Diallo M.\,S., Kulesh M., Holschneider M., Sherbaum F., Adler F.} Characterization of polarization attributes of seismic waves using continuous wavelet transforms // Geophysics.~--- 2006.~--- V.~71, N.~3.~--- P.~67--77.
%
%\bibitem{s3}  \textit{Lyons R.} Digital signal processing.~--- M.: Binom, 2006.~--- P.~361--369. --- (in Russian).
%
%\end{thebibliography}
%\selectlanguage{russian}

\begin{flushright}
{ \textit Received by the editor of Proc. of MIPT 03.05.2022.}
\end{flushright} 

%\section{Introduction}

\end{document}